\documentclass{article}

\usepackage[english]{babel}
\usepackage{authblk}
\usepackage[letterpaper,top=2cm,bottom=2cm,left=3cm,right=3cm,marginparwidth=1.75cm]{geometry}

\usepackage{csquotes}
\usepackage{amsmath, amsfonts, amsthm, amssymb}
\usepackage{graphicx}
\usepackage[colorlinks=true, allcolors=blue]{hyperref}
\usepackage{algorithm}
\usepackage{algpseudocode}
\usepackage{braket}
\usepackage[dvipsnames]{xcolor}
\usepackage{fontawesome5}
\usepackage{mathtools}

\usepackage[autocite=superscript,style=numeric,backend=bibtex,sortcites=true]{biblatex}
\newcommand{\orcid}[1]{$\,$\href{https://orcid.org/#1}{\textcolor[HTML]{A6CE39}{\faOrcid}}}
\DeclarePairedDelimiter\ceil{\lceil}{\rceil}

\theoremstyle{plain}
\newtheorem{theorem}{Theorem}
\newtheorem{conjecture}{Conjecture}

\theoremstyle{definition}
\newtheorem{definition}{Definition}[section]

\begin{document}

\title{Bounding the Graph Capacity with Quantum Mechanics and Finite Automata}
\author[1]{Alexander Meiburg\footnote{\href{mailto:shanqdfa@ohaithe.re}{shanqdfa@ohaithe.re}}\orcid{0000-0002-4506-9146}}
\affil[1]{Perimeter Institute of Theoretical Physics}
\affil[1]{Institute for Quantum Computing, University of Waterloo}

\maketitle

\begin{abstract}
The zero-error capacity of a channel (or ``Shannon capacity of a graph") quantifies how much information can be transmitted with no risk of error. In contrast to the Shannon capacity of a {\em channel}, the zero-error capacity has not even been shown to be computable: we have no convergent upper bounds. In this work, we present a new quantity, the zero-error {\em unitary} capacity, and show that it can be succinctly represented as the tensor product value of a quantum game. By studying the structure of finite automata, we show that the unitary capacity is within a controllable factor of the zero-error capacity. This allows new upper bounds through the sum-of-squares hierarchy, which converges to the commuting operator value of the game. Under the conjecture that the commuting operator and tensor product value of this game are equal, this would yield an algorithm for computing the zero-error capacity. 
\end{abstract}

\section{Introduction}
The Shannon capacity $\Theta$ of a graph $G$ measures the {\em zero-error capacity} of a communication channel with symbols $V(G)$ and a confusion graph $G$. Denote by $\alpha(G)$ the size of the largest independent set in $G$, and by $\boxtimes$ the strong graph product. $G^{\boxtimes n}$ is the $n$-fold strong graph product of $G$ with itself. Then the Shannon capacity is defined\supercite{Shannon56} as
\begin{equation}
    \Theta(G) = \lim_{n\to \infty} \sqrt[n]{\alpha\big(G^{\boxtimes n}\big)}.
\end{equation}
The limit always exists, as a consequence of Fekete's Lemma and the supermultiplicativity of $\alpha$. In terms of communication, $\Theta(G)$ counts asymptotically how many $n$-character strings with symbols in $V(G)$ can be all pairwise {\em distinguishable}: for any pair of strings $x$ and $y$, there is at least one position $x_i$ and $y_i$ that is not {\em confusable}, i.e. $(x_i,y_i) \not \in E(G)$.

The behavior of the Shannon capacity is difficult to control\supercite{Alon98,Alon06,Bi2019}. While every independent set of a $G^{\boxtimes n}$ yields a lower bound, and these lower bounds converge, it is difficult to identify when the sequence might converge\supercite{Alon06}. While it is true that
\begin{equation}
    \Theta(G \sqcup H) \ge \Theta(G) + \Theta(H)
\end{equation}
\begin{equation}
    \Theta(G \boxtimes H) \ge \Theta(G) \Theta(H)
\end{equation}
as one would expect, equality does {\em not} hold in general\supercite{Alon98,Schrijver23}. The Shannon capacity of the 5-cycle, $\Theta(C_5) \ge \sqrt{\alpha(G^{\boxtimes 2})} = \sqrt{5}$, was open for 24 years until Lovasz\supercite{Lovasz79} provided a matching upper bound of $\sqrt{5}$, a bound that today is often understood in terms of semidefinite programming. Other upper bound techniques\supercite{Haemers81,Hu18} have appeared, but even the capacity of the 7-cycle remains open. Recently, Bi \& Tang\supercite{Bi2019} showed that a large class of upper bounds based on conic programming cannot beat the bound of Lovasz. It is not even known whether graph capacity is computable.

In this paper we put forward a new upper bound on Shannon capacity using a Sum of Squares hierarchy. We start by relating Shannon capacity to what we call the {\em reversible capacity} of a graph. This quantity comes from a perspective less of graph theory, and more returning to the original setting of communication channels and error correcting languages: in particular, languages recognized by reversible DFAs\supercite{Lombardy02_revdfa,Pin92_revdfa}. We show that reversible capacity converges to the Shannon capacity in a controlled way, so that if the reversible capacity is computable, then the Shannon capacity is too (Theorem \ref{thm:comp_imply_comp}). Motivated by this, a further modification as the {\em unitary capacity} is equal to the Shannon capacity, and can be written as a tensor product Noncommutative Polynomial Optimization (NPO) problem \supercite{NPA08,BKP16}. This problem has convergent upper bounds to the commuting operator value, which is an upper bound to the true value that we conjecture to be tight.

If the conjecture is proved correct and our NPO has equal commuting operator and tensor product values, then this would give a computable algorithm for the zero-erro capacity of a graph.

\section{Graph Capacity and Regular Languages}
The graph capacity $\Theta$ counts tuples of vertices that form independent sets. This is thought of as the limit of sets of strings of successively longer lengths $n$: $\Theta(G) = \lim_{n\to\infty}\sqrt[n]{\alpha(G^{\boxtimes n})}$ requires counting strings of length $n$ that are all pairwise distinguishable (not confusable). But if we are considering the capacity of a channel described by $G$, then in an equally valid perspective, we can think of a {\em language} of acceptable strings, of many different lengths. For instance, for the cycle graph $C_5$, an optimal set of $n=1$ is the set of strings $\{0,2\}$. We can take the language of all strings built from these: $L_{C_5,2} = \{0,2,00,02,20,22,000,002,\dots\} = \{0,2\}^*$. Any two strings of the same length in this language are distinguishable, and so this language $L_{C_5,2}$ witnesses the fact that $\Theta(G) \ge 2$ because its {\em growth rate} as a language is 2. An optimal set for $n=2$ on the five-cycle is $\{00,12,24,31,43\}$, and from that we can build the language $L_{C_5,\sqrt{5}} = \{00,12,24,31,43\}^*$. Since this as a growth rate of $\sqrt{5}$, it witnesses that $\Theta(G) \ge \sqrt{5}$.

It is clear that the language-oriented perspective gives an equivalent definition of graph capacity, because at any finite $n$ achieving a capacity $\sqrt[n]{\alpha(G^n)}$, we can build a corresponding language with the same growth rate; and vice versa, by taking a language with a given growth rate, we can approach that growth arbitrarily closely by looking at sets of strings of increasing lengths. We formalize this as follows.

\begin{definition}[Graph languages]
A language $L$ over a graph $G$ is a language over the alphabet $V(G)$. The language $L$ is {\em distinguishable} if for any two strings $x,y\in L$ of equal length, there is some index $i$ such that $(x_i,y_i) \not\in E(G)$.
\end{definition}
\begin{definition}[Growth Rate]
The {\em growth rate}\footnote{This definition seems to be standard terminology, see e.g. \supercite{Shur10_growth}, but I am unable to find a source as to who actually coined it. Perhaps it is simply too natural.} of a language $L$ is $\lim\sup_{n\to \infty}|L(n)|^{1/n}$, where $L(n)$ is set of words in the language of length $n$ and so $|L(n)|$ is the number of words of length $n$.
\end{definition}

The two example languages above have a property in common: the languages are {\em regular}. That is to say, there is a finite state automaton that accepts exactly those languages; alternately, that they are given by a regular expression. The regular expressions are \verb?(02)*? and \verb?(00|12|24|31|43)*?. Regular languages have many nice properties and, in particular, it is easy to reason about their growth rate. So, perhaps instead of searching for independent sets, we can search for regular languages!

\begin{theorem}\label{thm:regCapEqCap}
The graph capacity among all languages (that is, the Shannon capacity) is equal to the capacity among regular languages.
\end{theorem}
\begin{proof}
For any $n$, we can take a set $S$ of length-$n$ strings of size $\alpha(G^{n})$. Then we make a regular language out of all repetitions of those strings, and this gives a regular language of growth rate $\sqrt[n]{\alpha(G^n)}$. Taking the supremum over all regular languages, we see that we can converge to the Shannon capacity.
\end{proof}

For any class \textsf{C} of languages, we could define the graph capacity $\Theta_\textsf{C}(G)$ of {\em over that language}, the largest growth rate of any language $L \in \textsf{C}$ that is a graph language for $G$. Then the zero-error capacity $\Theta$ is $\Theta_\textsf{ALL}(G)$, since we put not restriction on the language. The equivalent characterization of Shannon capacity by independent sets is $\Theta_{\textsf{BLOCK}}(G)$, where \textsf{BLOCK} is the family of block-coding languages: languages that are the closure under Kleene star of a fixed list of strings of equal length. If $L_1 \subseteq L_2$, then $\Theta_{L_1}(G) \le \Theta_{L_2}(G)$. Then Theorem \ref{thm:regCapEqCap} can be simply proved as
\begin{equation}\label{eqn:eqCapAllBlckReg}
\Theta_\textsf{ALL} = \Theta_\textsf{BLOCK} \le \Theta_\textsf{REG} \le \Theta_\textsf{ALL}
\end{equation}
as a simple consequence of $\textsf{BLOCK} \subset \textsf{REG} \subset \textsf{ALL}$.

So, we don't lose anything by searching among regular languages. The next convenient fact is that, given a regular language, we can efficiently check whether it is correct graph language (a valid code).

\begin{theorem}\label{thm:polycheck}
There is a polynomial-time algorithm to determine if a regular language $L$ (given as a finite automaton) is a graph language for $G$, that is, if $L$ accepts some pair strings which are confusable in the graph $G$.
\end{theorem}
\begin{proof}
We can identify a {\em pair} of strings of equal length with their {\em perfect shuffle}\footnote{This is another term which appears standard, but almost exclusively in the context of undergraduate homework problem sets - for instance \href{https://cs.stackexchange.com/questions/32621/is-context-free-language-closed-under-perfect-shuffle}{here}, \href{http://www.cs.nthu.edu.tw/~wkhon/assignments/assign1ans.pdf}{here}, or \href{https://cseweb.ucsd.edu/~rjestesj/docs/CSE105/hmwk2.pdf}{here}. In fact, of the two instances I could find in scholarly work, \supercite{Biegler09_pshuffle,DNJ09_pshuffle}, one of them focuses\supercite{DNJ09_pshuffle} on it only as a nice problem for teaching undergraduates.}: take each character from a string and alternate them. The strings $abc$ and $xyz$ become $axbycz$, for instance. If the original strings are from a regular language, the set of perfect shufflings is also a regular language, which can be seen by just running two finite automata in parallel. If we take the perfect shufflings of $L$ with itself, we get all pair of strings that $L$ accepts. We would want to compare this against all confusable pairs of strings. These are strings where each pair of corresponding characters is confusable: either equal or adjacent in $G$. For $C_5$, for instance, this is:
\begin{verbatim}
Lconfusable = (00|11|22|33|44|01|12|23|34|40)*
\end{verbatim}
In particular, we want to know if $L$ accepts any {\em non-equal} confusable pair of strings. Those are much like the above, except that at least one character somewhere in the middle must not be equal. For $C_5$, this is:
\begin{verbatim}
Lne = (00|11|22|33|44|01|12|23|34|40)*(01|12|23|34|40)(00|11|22|33|44|01|12|23|34|40)*
\end{verbatim}
and this will always be a regular language. More generally, we write \verb?Lne = (E|C)*C(E|C)*?, where \verb?C? is union of all non-equal confusable pairs, and \verb?E? is all the equal pairs. So, we construct this language $\texttt{Lne}$ of perfect shufflings of nonequal confusable strings. We compare to the perfect shuffle of our given language, $\textsf{PShuffle}(L,L)$. Since $\texttt{Lne}$ is regular and $L$ is regular, $\textsf{PShuffle}(L,L)\cap \texttt{Lne}$ is regular, and then we ask: $\textsf{PShuffle}(L,L)\cap \texttt{Lne}= \emptyset$? If it does, $L$ is a valid graph language; otherwise, $L$ accepts a bad pair of strings, and it is not a graph language for the graph $G$.
\end{proof}

Finally, the growth rate of $L$ is easy to compute, a folklore result that we only need to inspect the dominant eigenvalue of the transition function.
\begin{theorem}\label{thm:growth}
For a regular language $L$ on a language with $k$ symbols with a $d$-state DFA, there is a $O(kd^2+d^3)$ algorithm to find the growth rate of $L$.
\end{theorem}
\begin{proof}
The state transitions of $L$ can be represented as 0-1 array of size $d\times k\times d$ state transition matrix $M$: each pair $(s,x)$ transitions from a state $s$ and symbol $x$ to a new state $s'$ iff $M_{s,x,s'} = 1$. To be well-formed, we require that $\forall_{s,x} \sum_{s'} M_{s,x,s'} = 1$. To define the language, some of the states are marked as accepting, and the rest are rejecting; and one state is marked as initial. Then the {\em reduced matrix} $R \in \mathbb{N}^{d\times d}$ is $R_{s,s'} = \sum_x M_{s,x,s'}$, and counts the number of ways to transition from a state $s$ to $s'$.

First check whether the initial state can reach every other state: if it cannot, then remove the unreachable states from the DFA. Then take the largest eigenvalue $v$ of $R$ such that the corresponding right eigenvector $\mathbf{v}$ has nonzero support on the accepting states. The growth rate of $L$ is then $v$, that is, $|L|_n \sim v^n$. To see this, observe that the initial state will mix onto the eigenvector $\mathbf{v}$, this will grow like $v^n$, and then some nonzero fraction will lead to accepting states.
\end{proof}

This is enough, in principle, to start a search over DFAs of different state sizes: enumerate the different possible transition tables, check that they are valid, and compute their growth rate. Like brute-force searching independent sets on large $G^{\boxtimes n}$, this is guaranteed to converge to the true graph capacity from below, but possibly very slowly, and without any guarantees on the rate of convergence.

\section{Simplifying a DFA}\label{sec:simpDFA}
This section is concerned with properties and growth rates of DFAs. Although our ultimate goal will be dealing with reversible DFAs and then unitary automata, the analysis is very similar and understanding the DFA case first will make the later NPO construction much more natural.

\subsection{Connectedness}
There are a number of ways we can simplify our DFA without sacrificing capacity or correctness. First, some observations about absorbing states. {\em Absorbing states} are states of a DFA such that no symbol will ever make it change state; it is stuck there indefinitely. The strongly connected components of a DFA are the equivalence classes of states that can all reach one another.

\begin{theorem}\label{thm:one_abs}
Suppose $L$ is a regular language with a growth rate of $v$, that is a graph language for $G$, recognized by a DFA $D_1$ with $d$ many states. Then there is another DFA $D_2$ for a language $L_2$ with a growth rate $v$, also with at most $d$ states, with the additional properties:
\begin{itemize}
    \item $D_2$ has exactly one absorbing state $s_{abs}$.
    \item $s_{abs}$ is a rejecting state.
    \item The states besides $s_{abs}$ form one strongly connected component.
\end{itemize}
So we can impose these conditions on the DFAs we look for without sacrificing capacity.
\end{theorem}
\begin{proof}
First, we can assume that all states in $D_1$ are reachable. Otherwise, we can remove these states from $D_1$ without affecting it and only making the state space smaller.

Decompose the state space of $D_1$ into its strongly connected components, the collections of states that can all reach each other. The graph of transitions between components is a directed acyclic graph, so there must be some {\em final} vertex (a vertex with no edges leaving it) of the graph. This final component cannot have any accepting states. Suppose for contradiction that it did has some accepting state $s_1$. There is some string $x$ to bring us to the state $s_1$, because $D_1$ is entirely reachable. And for any $s_i$ in the same component as $s_1$, there is a string $y_i$ that brings us back to $s_1$, since it is strongly connected. Then for any string $z \in \Sigma^*$, $xz$ puts us in some state $s_i$ and $y_i$ brings us back to $s_1$, so the string $xzy_i$ is in the language $L$. This implies that $L$ has a growth rate of $|\Sigma| = |G|$, which is impossible since $G$ is nonempty. So, any final strongly connected components of $D_1$ must be entirely rejecting.

We can take all of the states in those final strongly connected components and replace them with one absorbing rejecting state. This is our $s_{abs}$. This modification affect the behavior of the DFA, since they were all only ever going to lead to rejecting states anyway. This shows that we can always ensure there is exactly one absorbing state, and it is rejecting.

This leaves the other strongly connected components. We claim we can assume there is only one. There must be some component $C$ that leads only to $s_{abs}$. (This is any component that would become a final vertex after removing $s_{abs}$ from the graph, and there is always a new final vertex.) If the $C$ component has only rejecting states, then it can be removed and all of its incoming transitions point straight to $s_{abs}$ instead, and we repeat. Otherwise, $C$ has some accepting states. If strings accepted from $C$ have a growth rate equal to $v$, then we can restrict to that subset: define $D_2$ to start from a state in $C$ and otherwise function as $D_1$. It will have growth rate $v$, and it will have just one connected component $C$ besides the absorbing state $s_{abs}$.

Otherwise, if $C$ has a growth rate less than $v$, we can replace all of $C$ by rejecting states and collapse it into $s_{abs}$ without affecting the growth rate of the DFA. We have reduced the number of connected components by one, and we can keep repeating until we find one of the appropriate growth.
\end{proof}

This connectedness requirement will helps us compute the growth rate of the DFA, as we'll show in the next section. The constraints of Theorem \ref{thm:one_abs} also imply that all states are reachable from the initial state, which will be useful in our NPO formulation. 

As an aside: if a computer search was performed over DFAs, we could trim our search space of allowed DFAs using these facts. Designating $s_{abs}$ as state number $d$, we already know all transitions out of it, and this leaves only searching over the $d-1$ others. The fact that the remaining states $S_{\neg abs}$ form a strongly connected component is equivalent to the statement that, for any bipartition $S_{\neg abs} = U \cup V$, with $U \cap V = \emptyset$, there is some $(u,s)$ that transitions to $v$. This is well-suited to integer programming techniques to enforce strong connectedness.

\subsection{Growth}
The formula for growth rate given in Theorem \ref{thm:growth} is a bit unwieldy, but for DFAs obeying the properties of Theorem \ref{thm:one_abs}, it becomes simpler. First, we don't need to ask whether all states are reachable, as we have already required it. There is one eigenvector for the reduced matrix $R$ supported entirely on $s_{abs}$, with an eigenvalue of $|G|$. It is apparent that there can be {\em no} other eigenvector with {\em any} support on $s_{abs}$, because acting on that eigenvector with $\frac{1}{|G|}R$ would only increase the mass on $s_{abs}$ and decrease the mass on the rest - except for one exception where there could be an eigenvector with eigenvalue zero. So, the other $d-1$ eigenvectors must each be supported entirely on $S_{\neg abs}$, or have a eigenvalue of zero. Since there is some positive eigenvalue in $S_{\neg abs}$, the maximal eigenvalue out of the $d-1$ others cannot be zero, and it has a maximal eigenvector supported entirely on $S_{\neg abs}$. Because the component is strongly connected, this maximal eigenvector has nonzero support on the whole component, and in particular has nonzero support on an accepting state.

So we have proved the following:
\begin{theorem}\label{thm:simple_growth}
For a DFA obeying the properties of Theorem \ref{thm:one_abs}, the growth rate is simply equal to the largest eigenvalue of $R_{\neg abs}$, the $(d-1)\times(d-1)$ reduced matrix on all states except $s_{abs}$.
\end{theorem}
In other words, the growth rate is just a maximal eigenvalue of the sum of the 0-1 reduced transition matrices. This is already suggestive of viewing these objects are linear operators more than just transition matrices. We can also add one additional property to our constraints:

\begin{theorem}\label{thm:one_acc}
For a DFA obeying the properties of Theorem \ref{thm:one_abs}, we can modify the DFA to have only one accepting state, without changing the state transition matrix or number of states, and without affecting the growth rate or correctness. Furthermore, the accepting state and initial state can be taken to be any state except the absorbing state - we can even require that they are equal.
\end{theorem}
\begin{proof}
Since the growth rate is equal to the maximal eigenvalue of $R_{\neg abs}$, the growth rate is independent of which states are accepting or rejecting - as long as there is at least one accepting state. So, if we have a DFA with multiple accepting states, we can turn all but one of them into rejecting states, and this won't affect the growth rate. This gives a new DFA with only one accepting state, but the same state transition matrix. Since we only accept a subset of the language that we did before, and all strings we accepted before weren't confusable with each other, the strings accepted by the new DFA also cannot be confusable, and it is still a valid independent set language.

To show that we can take the accepting state to be whichever want, remember that all states but the absorbing state are strongly connected. So if the DFA has initial (resp. accepting) state $s_i$, we can change this to $s_{i'}$. The new DFA still correctly recognizes only distinguishable strings, because if two accepted strings $x$ and $y$ are confusable, then there is some prefix (resp. suffix) string $z$ connecting $s_i$ to $s_{i'}$ so that $zx$ and $zy$ (resp. $xz$ and $yz$) would be confusable strings accepted in the original DFA.
\end{proof}

\subsection{Correctness}
We showed earlier that checking the correctness of a DFA can be done in polynomial time through the algebra of regular languages. Poking further, we can turn this into a local constrained-based formulation, one that is both more amenable to search as well as the basis for our eventual operator constraint formulation.

In Theorem \ref{thm:polycheck} we showed that checking is equivalent to running two copies of the DFA in parallel, the perfect shuffle, and asking whether they accept any language of the form \verb?(E|C)*C(E|C)*? -- where \verb?C? is union of all non-equal confusable pairs, and \verb?E? is all the equal pairs. By taking the middle \verb?C? to be the {\em first} place where two strings differ, we see that this is equivalent to the slightly simpler regex \verb?Lne' = E*C(E|C)*?. Asking whether the perfect shuffle DFA accepts any string in \verb?Lne'? is equivalent to asking whether there is a path on its $d^2$-sized state graph, starting from $(s_{init},s_{init})$, going through any number of \verb?E? edges, then a \verb?C? edge, and then any number of \verb?E? or \verb?C? edges, and ending up at $(s_{final,1},s_{final,2})$ of two (possibly distinct) accepting states. By doing a flood-fill on this graph, tracking which of the $d^2$ state pairs we can reach at each of these three steps, we can check whether such a path exists.

If our DFA follows the properties laid out in Theorems \ref{thm:one_abs} and \ref{thm:one_acc}, this becomes a bit simpler. First, we know there is only one accepting state, which is also the accepting state $(s_{init},s_{init})$. Second, we know that nothing from $(s_{abs},-)$ or $(-,s_{abs})$ could ever reach $(s_{init},s_{init})$, so we can just ignore $s_abs$ entirely. Any transition with two equal symbols (that is, \verb?E?) will map us from a diagonal state $(x,x)$ to another diagonal state $(y,y)$, so after moving from $(s_{init},s_{init})$ through \verb?E*? we only reach diagonal states. Since the nonabsorbing states are strongly connected, we can also reach {\em every} diagonal state. Since we can also end with as many symbols from \verb?E*? as we like, reaching any diagonal state $(x,x)$ near the end is sufficient to prove that we can reach $(s_{init},s_{init})$ at the very end.

This suggests the following algorithm for checking correctness of a DFA. We mark progressively more state pairs as $\textsf{Final}$ if they can reach some diagonal state $(x,x)$, until there are no more changes. The corresponding $\textsf{Initial}$ states are just the diagonal states. Then check if any edge $(u,v)$ in $G$ induces a transition from an $\textsf{Initial}$ state to a $\textsf{Final}$ state, and if so, the DFA is not a valid graph language. This is written in pseudocode in Algorithm \ref{alg:check}.

\begin{algorithm}
\caption{An algorithm for checking that a DFA $T$ accepts only a valid independent set language on a graph $G$}\label{alg:check}
\begin{algorithmic}
\Require $G$ of size $k$             \Comment{Graph}
\Require $T : [d]\times[k] \to [d]$ \Comment{State transition map}
\State $\textsf{Final} \gets \textsc{Identity}(d-1)$ \Comment{Starts with $1$'s on the diagonal, $0$ otherwise}
\While{$\textsf{Final}$ has changed}
    \For{$i\in [d-1]$, $j \in [d-1]$}
        \For{$u \in G$, $v \in G$}
            \If{$u=v$ or $(u,v)\in G$} \Comment{Propagate $E$ and $C$ respectively}
                \State $i' \gets T[i,u]$
                \State $j' \gets T[j,v]$
                \If{$\textsf{Final}[i',j']$} 
                    \State $\textsf{Final}[i,j] \gets 1$
                \EndIf
            \EndIf
        \EndFor
    \EndFor
\EndWhile
\For{$i\in [d-1]$, $(u,v) \in G$} \Comment{Now we check if any $\textsf{Initial}$ can reach $\textsf{Final}$}
    \State $i' \gets T[i,u]$
    \State $j' \gets T[i,v]$
    \If{$\textsf{Final}[i',j']$} 
        \State \textbf{Reject}
    \EndIf
\EndFor
\State \textbf{Accept}
\end{algorithmic}
\end{algorithm}

If we wanted to encode this algorithm in a constraint problem, this requires a variable for each of the $d^2$ variables in \textsf{Final}; this can be reduced to $(d-1)(d-2)/2$ variables by dropping the $s_{abs}$, the diagonal, and using the symmetry of the matrix. The propagation is encoded through 4-SAT clauses: if $u$ induces a transition from $i$ to $i'$, and $v$ induces a transition from $j$ to $j'$, and $\textsf{Final}[i',j']$ is true, then $\textsf{Final}[i,j]$ is true. This path finding description also implies that correctness can also be checked in the complexity class $\textsf{NL}$, nondeterministic logspace. This isn't deep -- in general recognizing where DFAs are empty or not is in $\textsf{NL}$, and here we're just applying that to $\textsf{PShuffle}(L,L)\cap \texttt{Lne}$.

\section{Operators}
The dream would be to drop the $d$ - to become agnostic to the size of the DFA. Then our state transitions are just operators on a $d$ dimensional Hilbert space obeying some relations. The facts that we deduce from the operator relations can give us bounds on the capacity. Our ultimate goal is to apply the NPA hierarchy\supercite{NPA08}, which permits polynomial constraints and gives a convergent hierarchy from above.

For DFAs, the operators $T_u$ are 0-1 matrices with exactly one ``1" per row. These are unfortunately hard to characterize directly. Eventually we will remedy this by switching to complex unitary operators and move past DFAs, but it is instructive to first look at how the operators will work for DFAs. For now, pretend we could {\em require} that $T_u$ are matrices of the correct form.

Our $d$ dimensional state space is now a vector space of unbounded dimension. Our state transition matrix associates to each vertex $u$ in the graph $G$ an operator $T_u$ on this space, describing how we move from one state to another. The reduced matrix $R$ is just the sum $\sum_u T_u$, and our objective function is the largest eigenvalue of $R_{d-1,d-1}$, or $P R P$ where $P$ is an appropriate projector ($P^2 = P$) that projects away the absorbing state(s) $s_{abs}$.

Things get complicated with tracking the correctness. We want to track the $\textsf{Final}$ states, but that data has $d^2$ many variables. If our optimization schema supported a notion of ``product space" we would be fine, but that's something we have to do ourselves. So we actually have to upgrade our whole optimization problem to be on a $d^2$-dimensional Hilbert space.

So now we have a ``left" and ``right" part of our space. Instead of having operators $T_u$, we have operators $T_{u,v}$ which we conceive as representing $T_u \otimes T_v$: our two parallel copies of the DFA read symbol $u$ and $v$ simultaneously. We can also define $T_{u,0} = T_u \otimes I$ and $T_{0,v} = I \otimes T_v$. We have commutation relations that $[T_{u,0}, T_{0,v}] = 0$ as constraints. We have a swap operator $S$ that switches these two parts of our space: $S^2 = I$, $S T_{u,v} = T_{v,u} S$. ($S$ is also Hermitian, unlike the $T$'s). We can either expand $T_{u,v} = T_{u,0} T_{0,v}$ in each constraint and use only $O(|G|)$ variables, or we can make the $T_{u,v}$ their own variables, reducing the polynomial order of other constraints, and we add constraints $T_{u,v} = T_{u,0} T_{0,v}$.

Next we want a projector $F$ on this larger space that marks which of the states are marked in $\textsf{Final}$ or not. It is populated with zeros and ones in our standard $d^2$ dimensional basis, so it is a projector on our larger space: $F^2 = F$. Symmetry says $[S,F] = 0$. Since we've already removed the $s_{abs}$ part of the space, everything on the ``diagonal" of $F$ must be 1. This is not the ``diagonal" in the sense that the identinty $I$ is diagonal, but of the basis vectors of the form $(u,u) = u \otimes u$. This set of diagonal states is also our $\textsf{Initial}$ states, so let's define a projector for all the diagonal states: $D$, with $D^2 = D$ and $[S,D]=0$. $F$ contains all of the $D$ states: $F \succeq D$. $D$ also preserves membership under any diagonal evolution, so $D T_{u,u} D = T_{u,u} D$.

Since $F$ is defined as the states that can reach a diagonal state, it is equivalently the closure of $D$ under ``evolves under confusable transitions". In other words, it preserves membership under action by $T_{u,u}$ or by $T_{u,v}$ for $(u,v)\in E(G)$. In terms of operator equations, $\forall_{u}: F T_{u,u} F = T_{u,u} F$, and $\forall_{(u,v)\in E(G)} F T_{u,v} F = T_{u,v} F$. The correctness is specified by the constraint that an edge in $G$ can't mix an initial state (a diagonal state) into a final state: $\forall_{(u,v) \in E(G)} F T_{u,v} D = 0$.

To make the objective easier to work with, we drop the state $s_{abs}$ from our state space entirely, and we just describe the DFA as one that transitions within a single component, and sometimes fails early, a structure formally referred to as a {\em partial DFA}. Then all states are eventually accepting, the choice of initial and accepting states are irrelevant, and $R = \sum T_u$. The requirement that $T_u$ has {\em exactly} one ``1" per row becomes {\em at most} one ``1".

But we have moved to the tensor product space, and we can no longer write the objective as $\lambda_{\max}(R) = \lambda_{\max}(\sum_u T_u)$, since we have no $T_u$ variables. We could instead use $R\otimes R = \sum_{u,v} T_{u,v}$, and then our growth rate is the square root of $\lambda_{\max}(\sum_{u,v} T_{u,v})$. Or we could just the left (or right) space, $R \otimes I$, and our growth rate is $\lambda_{\max}(\sum_u T_{u,0})$. Or we can use the diagonal space, $\lambda_{\max}(\sum_u T_{u,u})$. All of these would give correct bounds, under the (admittedly unrealistic) assumption that we can force $T_{u,u}$ to all be 0-1 matrices. Note that in general these would not be equal for general operators, but for these 0-1 matrices where each row has sum 0 or 1, they are equal.

One issue with that objective function is that we could always set $D = F = 0$ and trivially satisfy all of the constraints. We can address this is by making their nonzero-ness part of the objective function. With an objective function $\lambda_{\max}(\sum_u T_{u,u})$, we can project with $D$ on either side: $\lambda_{\max}\left(D \left(\sum_u T_{u,u}\right) D\right)$. We've already required that $D^2 = D$, and now we're saying that we only care about the growth rate of $T_{u,u}$ on the support of $D$. Again, this doesn't change the $\lambda_{\max}$ for the matrices we care about, but if we allowed matrices with other elements than 0 or 1 (e.g. stochastic matrices), this {\em would} change the value of the objective. This fixes the $D=F=0$ issue, as we can always move to the DFA defined by the restriction to the support of $D$, and this would give a new valid DFA with the same objective, and where $D$ is indeed the full diagonal.

This would give a relaxation of the problem, except that we {\em cannot} characterize 0-1 transition matrices as operators. Which is why we'll need to move on to {\em reversible} languages.

\section{Reversibility}
As noted in the previous section, ideally our $T_u$'s should 0-1 matrices with at most 1 per row, but this is very difficult to enforce. Our end goal is unitary matrices, which are nicely characterized by $UU^\dagger = U^\dagger U = I$. But not all 0-1 matrices are unitary, in fact only permutation matrices are. In order to facilitate this path, we need to digress and discuss {\em reversible} DFAs\supercite{Pin92_revdfa,Lombardy02_revdfa}. Except that, to make things slightly more complicated, we don't we actually want unitary transition matrices, because of our trick of dropping the absorbing state. Accordingly, we will ask for a {\em reversible, partial} DFA on the remaining states.

\begin{definition}{\cite{HJK15_revdfa}}
A {\em partial DFA} if for every symbol $s$ has an associated 0-1 transition matrix $T_s$, with each row sum at most 1 (as opposed to a standard DFA, which has a row sum of exactly 1.) Transitions with a 0 row correspond to immediately rejecting. A partial DFA is {\em reversible} if each matrix $T_s$ also has a {\em column} sum of at most 1 for each column.
\end{definition}

If the DFA was total (i.e. not partial), this would be equivalent to saying that each $T_s$ is a permutation matrix. Partial reversible DFAs can be characterized as precisely those whose transition matrices are {\em subpermutation matrices}. Reversible DFAs are well-studied. It is known that not all DFAs have a corresponding reversible DFA; reversible DFAs only recognize a subset of the regular languages. The question is now whether reversible regular languages achieve the full Shannon capacity.

The {\em reversible capacity} is then $\Theta_{\textsf{REV}}(G)$ of the graph $G$ is then the supremum growth rate over all reversible DFAs that recognize graph languages for $G$. The corresponding matrices obey $T^\dagger T \preceq I$, and a slightly stronger condition, $(T^\dagger T)^2 = T^\dagger T$. The objective is then $D\sum_g T_{g,g} D$, which isn't Hermitian, but doesn't itself preclude a Sum-of-Squares formulation (see Appendix \ref{sec:nonHermSoS}). Given the above noncommutative operator characterization, the reversible capacity can be bounded to some degree by the Sum-of-Squares hierarchy, but there doesn't seem a strong reason to believe this formulation is tight. It is clear that $\Theta_{\textsf{REV}}(G) \le \Theta(G)$. What more can we say about this quantity?

\subsection{Properties of the Reversible Capacity}
\begin{theorem}\begin{equation}
    \Theta_{\textsf{REV}}(G) \ge \alpha(G).
\end{equation}\end{theorem}
The reversible capacity is always at least as large as the largest independent set in $G$. To see this, fix a maximum independent set $S \subseteq G$. Take the one-state DFA, that accepts symbols in $S$ and rejects everything else. Since it is only one state, it is reversible, and it clearly has capacity $\alpha(G)$ and recognizes an independent set language.
\begin{theorem}\begin{equation}
    \Theta_{\textsf{REV}}(G \cup H) \ge \Theta_{\textsf{REV}}(G) + \Theta_{\textsf{REV}}(H).
\end{equation}\end{theorem}
where $\cup$ denotes the disjoint union of the graphs. The analogous statement holds for the Shannon capacity, and Shannon conjectured that in fact equality always held; but this was disproved by Alon\supercite{Alon98}. In our case, take a reversible DFA for $G$ and another for $H$. Create a new DFA that runs these two in parallel, updating the $G$ half of the state whenever it encounters a vertex from $G$, and the $H$ half whenever it encounters a vertex from $H$. The transition matrices are now exactly the Kronecker products $T_g \otimes I$ and $I \otimes T_h$. Because these operators commute, we have
$$\lambda_{\max}\left( \sum_g T_g \otimes I + \sum_h I \otimes T_h\right) = \lambda_{\max}\left( \sum_g T_g \otimes I\right) + \lambda_{\max}\left(\sum_h I \otimes T_h\right) = \Theta_{\textsf{REV}}(G) + \Theta_{\textsf{REV}}(H).$$
\begin{theorem}\begin{equation}
    \Theta_{\textsf{REV}}(G \boxtimes H) \ge \Theta_{\textsf{REV}}(G) \Theta_{\textsf{REV}}(H).
\end{equation}\end{theorem}
This is the equation that we would like to be an equality, and if it was, would imply that $\Theta_R = \Theta$. Again, to see that this is true, run a DFA for $G$ and $H$ in parallel. Now, each vertex is labelled as a pair $(g,h)$, and the transition process each separately as $T_g \otimes T_h$. Then
$$\lambda_{\max}\left( \sum_{g,h} T_g \otimes T_h\right) = \lambda_{\max}\left( \sum_g T_g \otimes I\right) \lambda_{\max}\left(\sum_h I \otimes T_h\right) = \Theta_{\textsf{REV}}(G) \Theta_{\textsf{REV}}(H).$$

We can also show that $\Theta_{\textsf{REV}}$ is close to $\Theta$ when the capacity of a graph is much smaller than the graph itself. For non-reversible DFAs, we have the simple language that takes an independent set $\alpha(G^n)$, looks at symbols $n$-at-a-time, and accepts strings in the independent set. This doesn't give us reversible DFAs in general. But we can make it reversible at the expense of some extra time. Taking an $n$ symbol independent set of $G^n$, we give our DFA one state for each element in this independent set. Now we want to reversibly rewind this back to the initial state. If we number the independent set $1\dots M$ with $M = \alpha(G^n)$, our DFA can insist that after the $n$-symbol independent set $m$, it needs to see the number $m$ written in a $|G|$-ary representation. This can't lead to any new confusable strings, because the $M$ independent sets were already distinguishable, and now we're just adding a suffix. Since we write a number from 1 to $M$ using digits in base $|G|$, we need $\ceil{\log_{|G|}(M)}$ many digits to write out.

So we need a total length of
$$n + \ceil{\log_{|G|}(M)} \le n + \log_{|G|}(M) + 1$$
to code for $M$ symbols, giving capacity at least
$$\Theta_{\textsf{REV}} \ge M^{1/(n + \log_{|G|}(M) + 1)}$$
For large $n$, $M = O(\Theta^n) \ge c\Theta^n$ for some fixed $c$, so this expression becomes
$$\Theta_{\textsf{REV}} \ge c^{1/(n + \log_{|G|}(\Theta^n) + 1)}\Theta^{n/(n + \log_{|G|}(\Theta^n) + 1)} \overset{n}{\longrightarrow} \Theta^{n/(n + n\log_{|G|}(\Theta) + 1)} =$$
$$\Theta^{1/(1 + \log(\Theta)/\log(|G|) + 1/n)} \overset{n}{\longrightarrow} \Theta^{\log(|G|)/(\log(|G|)+\log(\Theta))}.$$
So,
\begin{theorem}The reversible capacity is bounded by
\begin{equation}
    \Theta_{\textsf{REV}}(G) \ge \Theta(G)^{\log(|G|)/(\log(|G|)+\log(\Theta(G)))}
\end{equation}
or slightly differently,
\begin{equation}\label{eqn:thr_bound_ge}
    \log \Theta_{\textsf{REV}} \ge \frac{\log |G|\log \Theta }{\log |G| + \log \Theta}.
\end{equation}\end{theorem}
In particular, for a family of graphs $( G_i )$ with $\log \Theta(G_i) = o(\sqrt{\log V_i})$, we know that $\lim \frac{\Theta_R}{\Theta} = 1$. This is enough to prove a very useful bound,

\begin{theorem}\label{thm:comp_imply_comp}
The Shannon capacity is within bounded error of the reversible capacity by
\begin{equation}\label{eqn:thr_tht_bound}
    \Theta_{\textsf{REV}}(G) \le \Theta(G) \le \Theta_{\textsf{REV}}(G)^{\log |G| / (\log |G| - \log \Theta_{\textsf{REV}}(G))}.
\end{equation}
If there is an algorithm to give convergent upper bounds on the reversible capacity, then the Shannon capacity is also computable.
\end{theorem}
\begin{proof}
The first equation is given by solving Eq. \ref{eqn:thr_bound_ge} for $\Theta(G)$ to get one side, and the fact that $\Theta_{\textsf{REV}}(G) \le \Theta(G)$ on the other.

$\Theta(G)$ will be computable if we can provide convergent upper bounds on it. Fixing $G$, consider the sequence of graphs $G_i := G \nabla K_i$, where $G \nabla H$ denote the graph join operation that connects all vertices in $G$ to all vertices in $H$, and $K_i$ is the $i$-vertex complete graph. Any independent set of $G_i^{\boxtimes k}$ can be changed to one of equal size that does not use any coordinates from the $K_i$ within $G_i$. So $\alpha(G_i^{\boxtimes k}) = \alpha(G^{\boxtimes k})$, and we see that $\Theta(G_i)=\Theta(G)$. If we wish to bound $\Theta(G)$ to a given accuracy, we just need to take $i$ sufficiently large and get sufficiently good upper bounds on $\Theta_{\textsf{REV}}(G_i)$, and then Eqn. \ref{eqn:thr_tht_bound} puts turns that into good upper bounds on $\Theta(G)$.
\end{proof}
To make the above relationship more concrete, if we wish to bound $\Theta(G)$ to within a multiplicative accuracy $1+\epsilon$, then it suffices when
$$\log |G_i| \ge \frac{\log(\Theta_{\textsf{REV}}(G_i)) \left(\log(1+\epsilon)+\log(\Theta_{\textsf{REV}}(G_i))\right)}{\log(1+\epsilon)}$$
and then applying the bound $\Theta_{\textsf{REV}}(G_i) \le \Theta(G) \le |G|$,
$$|G_i| \ge |G|^{1+\frac{\log(|G|)}{\log(1+\epsilon)}} \approx |G|^{1+\frac{1}{\epsilon}\log(|G|)}$$
is sufficiently large of a graph on which to upper bound $\Theta_{\textsf{REV}}$.

We computations on some small graphs in Appendix \ref{app:SmallRev}.

\section{Quantizing}
Reversible capacity is nice because we can bound the error between it and the true capacity. But reversible transition matrices still involve the fact that their entries are 0 or 1, a requirement that we cannot enforce in the NPO framework. So, we need to upgrade our model again to a new set of languages: Nondeterministic Quantum Languages, or \textsf{NQL}.

\textsf{NQL}, or the Kondacs-Watrous Quantum Finite Automata\supercite{KW97_NQL}, are defined by a set of unitaries $U_s$ on a finite-dimensional Hilbert space, where $s$ ranges over the alphabet of the input. There are designated projectors $A$ and $R$ (with disjoint support) that project onto an accepting and rejecting subspace, respectively. The remaining non-terminating states $N = I-A-R$. After each symbol is consumed by the automaton, a projective measurement is performed to check if the automaton is in $A$, $R$, or $N$ respectively. If it's one of the former two, the string is accepted (resp. rejected) early and processing stops. Otherwise processing continues. And if the whole string is accepted with any {\em nonzero} probability, we say that the string is accepted by the automaton.

The class \textsf{NQL} in fact contains all regular languages\supercite{KW97_NQL,YS10_NQL} despite the unitarity requirement, which in turn means that the corresponding capacity $\Theta_{\textsf{NQL}}$ will be at least as great as the regular capacity $\Theta_{\textsf{REG}}$. And since $\Theta_{\textsf{REG}}\le\Theta_{\textsf{NQL}}\le\Theta_{\textsf{ALL}}=\Theta_{\textsf{REG}}$, this \textsf{NQL} capacity is exactly the zero-error capacity.

But then we need {\em one more} change. In the NPO framework, it is difficult to count strings ``accepted with some nonzero probability"; such non-equality facts are hard to deal with. So we define the {\em unitary capacity} $\Theta_U$ as the supremum over Kondacs-Watrous Quantum Finite Automata that correctly accept a graph language, but when computing the growth rate, we only count strings proportional to their acceptance probability.

In other words, take a quantum finite automaton $\{U_s\}$. It defines a language in \textsf{NQL}. We require that this language is a valid code -- no two accepted strings are confusable. We define the {\em capacity} of that automaton by
$$\textsc{Capacity}(\{U_s\}) = \lim_{n\to \infty}\left(\sum_{S\in \Sigma^n}\Big|A\prod_{s\in S} (U_sN) \ket{init}\Big|^2\right)^{1/n}$$
This is distinct from the growth rate of the corresponding language,
$$\textsc{GrowthRate}(\{U_s\}) = \lim_{n\to \infty}\left(\sum_{S\in \Sigma^n}sign\left(\Big|A\prod_{s\in S}(U_sN) \ket{init}\Big|^2\right)\right)^{1/n}$$
which counts each nondeterministically accepted string with equal weight of one. Which leads to our definition,
\begin{definition}
    For a graph $G$, the {\em unitary capacity} $\Theta_U(G)$ is the supremum over all Kondacs-Watrous Quantum Finite Automata $\{U_s\}$, that correctly define a distinguishable language on $G$, of the capacity $\textsc{Capacity}(\{U_s\})$.
\end{definition}
This will be the quantity which our NPO ultimately corresponds to.

\subsection{Unitary Capacity}
The unitary capacity is well-behaved enough for our needs.
\begin{theorem}\label{thm:thtU_lt_tht}
The unitary capacity $\Theta_U(G)$ is at most the zero-error capacity $\Theta(G)$.
\end{theorem}
\begin{proof}
The unitary capacity is defined as a supremum over automata that define true graph languages, languages in \textsf{NQL}. On the other hand, it can only undercount the growth rate of the language, so we are taking the supremum over values that are at most capacities of valid languages.
\end{proof}
\begin{theorem}\label{thm:thtU_gt_thtRev}
The unitary capacity $\Theta_U(G)$ is at least the reversible capacity $\Theta_{\textsf{REV}}(G)$.
\end{theorem}
\begin{proof}
Every reversible partial DFA gives a corresponding quantum automaton, one which never creates superpositions. Undefined transitions (the {\em partial} aspect of the DFA) becomes determinsitic transitions into the $R$ subspace. The reversibility guarantees that all these transitions are unitary. All accepted strings are accepted with probability 1, so for this resulting set of unitaries,
$$\textsc{Capacity}(\{U_s\}) = \textsc{GrowthRate}(\{U_s\})$$
and is equal to the growth rate of the reversible DFA. Since we can achieve the capacity of any reversible automaton, the unitary capacity is at least as large.
\end{proof}

Theorems \ref{thm:thtU_lt_tht} and \ref{thm:thtU_gt_thtRev} mean that the unitary capacity is safely sandwiched between the reversible and Shannon capacity. Combining this with Theorem \ref{thm:comp_imply_comp}, we see that getting bounds on the unitary capacity will suffice.

\subsection{Quantizing Transitions}\label{sec:quantize1}
To write this as an NPO, our quantum DFA is defined by its isometries $U_{i}$, so we take the left and right variables $U_{i,0}$ and $U_{0,j}$. We require that they commute, $U_{i,0}U_{0,j} = U_{0,j}U_{i,0}$, and build their products $U_{i,j} := U_{i,0}U_{0,j}$. The requirement that they are isometries reads $0 \preceq U_{i,0} U_{0,i}^\dagger \preceq I$.

Now evolution of the DFA goes by unitary conjugation, instead of matrix multiplication. That is, instead of $T_{0,i} D$, we get $U_{0,i} D U_{0,i}^\dagger$. The diagonal and final states are closed under evolution by the same pair $(g,g)$, and final states are also closed under confusable transitions:
\begin{align}\label{eqn:closure}
    \forall_{g},\quad D U_{g,g} D U_{g,g}^\dagger D = U_{g,g} D U_{g,g}^\dagger\\
    \forall_{g},\quad F U_{g,g} F U_{g,g}^\dagger F = U_{g,g} F U_{g,g}^\dagger\\
    \forall_{(g,h)\in E(G)},\quad F U_{g,h} F U_{g,h}^\dagger F = U_{g,h} F U_{g,h}^\dagger
\end{align}
where $F$ and $D$ are Hermitian variables $F = F^\dagger$, $D = D^\dagger$. The fact that all diagonal states are also final states still reads
\begin{equation}\label{eqn:PsuccD}
    F \succeq D.
\end{equation}
As before, the swap operator $S$ obeys
\begin{align}\label{eqn:Sprops}
    S^2 = 1,\quad S = S^\dagger\\
    SU_{i,0}S = U_{0,i}.
\end{align}
The diagonal states are unable to reach a final state by a confusable transition:
\begin{equation}\label{eqn:Pconfuse}
    \forall{(g,h)\in E(G)},\quad F U_{g,h} D U_{g,h}^\dagger F = 0
\end{equation}
Under this description, it's entirely possible that some unitaries lead to a state with a norm that is neither 0 (rejection) or 1 (accepting) but somewhere in between, with probability $p$. We can think of this as the QDFA sometimes erroneously rejecting with probability $1-p$. The important thing is that we never accidentally accept two confusable strings: this is guaranteed by the above equation, which is that the projection is {\em exactly} zero.

\subsection{Quantizing the Objective}\label{sec:quantize2}
The total evolution of the automaton is given by a sum of unitary conjugations, a channel $C$ with
$$C(\rho) \sum_g U_g^\dagger\rho U_g$$
Then $C$ is a linear map on the space of density matrices, and the capacity want to measure is the dominant eigenvalue of $C$. Since this can be equivalently written as $U\otimes U$ under changing the direction of the index, the objective is simply
$$\lambda_{\max}\left(\sum_g U_{g,g}\right)$$
Finally, we want to ensure that the objective actually only cares about the action of the $U$'s on the subspace we're checking -- that supported by $D$ -- so the objective in our NPO is
\begin{equation}\label{eqn:qobj}
\lambda_{\max}\left(D \sum_g U_{g,g} D\right)
\end{equation}

\subsection{Soundness, Completeness in Tensor Value}
In the above NPO problem, suppose that we also allow a {\em tensor product constraint}, a-la quantum game's tensor values: instead of just requiring that the $U_{g,0}$'s commute with the $U_{0,h}$'s, we require that the Hilbert space has a tensor product structure $A\otimes B$ so that the $U_{g,0}$'s only operate on $A$ and vice versa. This could be viewed as an additional constraint in the NPO framework -- but, admittedly, one that no longer gives a convergent relaxation hierarchy anymore. We can show, however, that our {\em tensor product NPO} is a complete and sound characterization of unitary capacity.

\subsubsection{Completeness}
To see completeness, realize that any reversible DFA $\mathcal{M}$ obeying the properties given in Section \ref{sec:simpDFA} can be embedded as unitaries in the manner prescribed. If the DFA has $k$-states, then the Hilbert space is $\mathbb{C}^k \otimes \mathbb{C}^k$. The unitaries $U_{g,0}$ are subpermutation matrices operating on the left factor with the state transition table of $\mathcal{M}$ induced by the symbol $g$; respectively $U_{0,g}$ operates on the right side. $U_{g,h}$ are the products.

The matrix $S$ is the \textsf{SWAP} operator that exchanges the left and right subspaces. $D$ is the projector onto the diagonal subspace, the space spanned by operators $\ket{i}\otimes \ket{i}$ for $i \in [k]$. And $F$ is the projector onto the subspace spanned by states $\ket{i}\otimes \ket{j}$, such that a pair of machines $\mathcal{M}\times \mathcal{M}$ can reach the state $(i,j)$ through a pair of unequal confusable strings.

The equations \ref{eqn:closure}-\ref{eqn:Pconfuse} then hold, under our assumption that $\mathcal{M}$ recognizes a valid code. The objective value \ref{eqn:qobj} is then the growth rate of the DFA, which means that our tensor product NPO achieves at least the capacity of $\mathcal{M}$. This proves completeness.

\subsubsection{Soundness}
To prove completeness, we need to show that given a set of operators fulfilling all equations, there is an actual code achieving a capacity at least as high as the objective of the NPO problem. Given a set of satisfying operators: first we define a set of accepted code words, then we show that this forms a valid code (no confusable codewords), and that it has the correct capacity.

Under the combination of all $U$'s, there is some minimal invariant subspace. This is the unitary analog of the strongly connected component for the DFAs, and we can restrict our Hilbert space to just this invariant subspace. In the event that there are multiple minimal invariant subspaces, take any that attains the maximum eigenvalue of the objective -- one of them must. The restriction of all operators to this subspace still forms a valid tensor product NPO assignment with the same objective, so we assume WLOG that this is our Hilbert space.

The operator $D$ is intended to be a projector onto the diagonal states $\ket{i}\ket{i}$. It must be a projector; this implied by our constraint $D^2 = D$. It also cannot be equal to the zero operator, because this would imply that the objective is equal to zero. So, it has a nonempty eigenspace of eigenvalue 1. 

Furthermore, it commutes with the swap operator $S$ (with $S^2 = I$). They must be simultaneously diagonalizable, so $D$ has a eigenspace of eigenvalue 1 that is also preserved by $S$ with an eigenvalue of $\pm 1$. We can assume WLOG the the eigenvalue under $S$ is $+1$, otherwise we can rewrite $S \to -S$ and get a new valid assignment to the NPO problem.

Pick an arbitrary unit vector in this space, $\ket{\psi}$. Then our code words $L$ are defined by
$$s \in L \quad \iff \quad \bra{\psi} \left(\prod_{i} U_{s_i,s_i}\right)\ket{\psi} > 0$$
That is, the string $s$ is accepted precisely when feeding $s$ to both copies of the DFA has some probability of returning to the initial state $\ket{\psi}$.

We need to show that this is a valid code. Suppose by contradiction that a pair of confusable strings $w$ and $x$ are both accepted. Write them in the form $w = pgy$, $x = phz$, where the shared prefix $p$ is as long as possible, and $g$ and $h$ are distinct and confusable symbols. The initial state $\ket{\psi}$ is in the support of $D$ and so also that of $F$. The conditions Eq.~\ref{eqn:closure},\ref{eqn:PsuccD} means that the later state
$$\ket{\psi_{yz}} = \left(U_{g,h}\prod_{s\in p} U_{s,s}\right)\ket{\psi}$$
will also be in the support of $F$; condition Eq.~\ref{eqn:Pconfuse} means it is completely orthogonal to $D$. Finally, conditions Eq.~\ref{eqn:Pconfuse}-\ref{eqn:Pconfuse} mean that this will stay into the $D$-orthogonal subspace to $F$ for the remaining confusable strings $y$ and $z$.

Since the final state is fully in $F-D$, and the initial state was in $D$, they must be orthogonal, and we have exactly zero overlap. So, any valid setting of unitaries leads to a valid code with the correct growth rate.

\subsection{Commuting vs. Tensor Value}
This gives a noncommutative polynomial problem whose tensor-product value is equal to the unitary capacity, so that if we could get tight bounds on this value, we could compute the Shannon capacity.

In general, the tensor-product value is not computable.\supercite{mip_re} The sum-of-squares hierarchy converges to the commuting value of the problem, for which lower bounds are not computable; searching through finite-dimensional values converges to the tensor-product value from below. But for many games of interest, the tensor-product and commuting operator values coincide. It feels, at least to the author, as though this should be one of those cases.
\begin{conjecture}
    The tensor-product and commuting operator values of the problem defined in Sections \ref{sec:quantize1} and \ref{sec:quantize2} are equal.
\end{conjecture}
and if this is true, this would mean the Shannon capacity is in fact computable. This conjecture would most likely be proven by modifying the soundness proof above to no longer require the tensor product structure.

\printbibliography

\appendix

\section{Small Cases of Reversible Capacity}\label{app:SmallRev}
It's interesting to ask for small graphs whether $\Theta_{\textsf{REV}}$ and $\Theta$ really do coincide. Since $\alpha(G) \le \Theta_{\textsf{REV}}(G) \le \Theta(G)$, so for the graphs where $\Theta(G)$ and $\alpha(G)$ coincide, so does $\Theta_{\textsf{REV}}(G)$. This already solves many small graphs. The first two nontrivial cases are the 5- and 7-cycles.

\begin{theorem} The reversible capacity is tight on the 5-cycle.
    $$\Theta_{\textsf{REV}}(C_5) = \sqrt{5}.$$
\end{theorem}
\begin{proof}
The easy independent set on $C_5^{\boxtimes 2}$ is $\{(1,1),(2,3),(3,5),(4,2),(5,4)\}$. The standard DFA to recognize patterns of this repeated is already reversible: the only memory is what the first symbol in each pair was, and there is no pair of pairs $(i_1,j)$ and $(i_2,j)$ that would lead to an irreversible transition induced by a symbol $j$. This shows $\Theta_{\textsf{REV}}(C_5) \ge \sqrt{5}$, and since $\Theta(C_5) = \sqrt{5}$, it is tight.
\end{proof}
The Shannon capacity of the 7-cycle is still open, but the reversible capacity seems reasonable capable.
\begin{theorem} The reversible capacity does better than $\alpha(G)$ on the 7-cycle.
    $$\Theta_{\textsf{REV}}(C_7) \ge \sqrt{10} = C_7^{\boxtimes 2}.$$
\end{theorem}
\begin{proof}
One independent set on $C_7^{\boxtimes 2}$ is $\{(0, 2), (0, 4), (1, 6), (2, 1), (2, 3), (3, 5), (4, 1), (4, 3), (5,  5), (7, 0)\}$. This is also recognized reversibly: the first symbol read puts the DFA in one of the five states $S_0, S_1, S_{24}, S_{35}$, or $S_7$. In each case, the accepted following symbols are $\{2,4\}, \{6\}, \{1,3\}, \{5\}$, or $\{0\}$, respectively, and this returns to the initial state. Since these five sets of allowed following symbols are disjoint, the transitions are reversible.
\end{proof}
An informative check would be, through explicit enumeration, checking for $n$ and $m$ if we can reversibly recognize maximum independent sets of $C_n^{\boxtimes m}$. A computer search for $\alpha_{C_7^{\boxtimes 3}} = 33$ should be feasible with a SAT solver.

\section{Non-commutative Polynomial Constraints for Reversible Capacity}\label{sec:revNPO}
For a graph $G$, this lists the full set of constraints (that we propose) for a reversible DFA on $G$.
\begin{align*}
\textrm{Variables:} &\quad S, \quad D,\quad F,\quad \forall_{g\in G}: T_{g,0}, T_{0,g},\quad \forall_{g,h\in G}: T_{g,h}\\
\textrm{Maximize:} &\quad \lambda_{\max}\left(D\left(\sum_{g} \frac{T_{g,g} + T_{g,g}^\dagger}{2}\right)D\right)\\
\textrm{Subject to:}
&\quad \forall_{g,h}\,\, T_{0,h} T_{g,0} = T_{g,0} T_{0,h}\\
&\quad \forall_{g,h}\,\, T_{g,h} = T_{g,0} T_{0,h}\\
&\quad \forall_{g}\,\, 0 \preceq T_{g,0}^\dagger T_{g,0} \preceq I,\quad 0 \preceq T_{0,g}^\dagger T_{0,g} \preceq I\\
&\quad \forall_{g}\,\, 0 \preceq T_{g,0} T_{g,0}^\dagger \preceq I,\quad 0 \preceq T_{0,g} T_{0,g}^\dagger \preceq I\\
&\quad \forall_{g}\,\, T_{g,0} T_{g,0}^\dagger T_{g,0} = T_{g,0},\quad T_{0,g} T_{0,g}^\dagger T_{0,g} = T_{0,g}\\
&\quad S = S^\dagger,\quad D = D^\dagger,\quad F = F^\dagger\\
&\quad S^2 = I,\quad D^2 = D,\quad F^2 = F\\
&\quad \forall_g\,\, S T_{g,0} = T_{0,g} S\\
&\quad S D = D S,\quad S F = F S,\quad D F = F D\\
&\quad \forall_{g}\,\, D T_{g,g} D = T_{g,g} D\\
&\quad F \succeq D\\
&\quad \forall_{g}\,\, F T_{g,g} F = T_{g,g} F\\
&\quad \forall_{(g,h)\in E(G)}\,\, F T_{g,h} F = T_{g,h} F\\
&\quad \forall_{(g,h)\in E(G)}\,\, F T_{g,h} D = 0\\
\end{align*}
This last equation is redundant with the above, but could help an actual computational attempt:
\begin{align*}
&\quad \forall_{g,h}\,\, S T_{g,h} = T_{h,g} S\\
\end{align*}

\section{Non-Hermitian Sum-of-Squares}\label{sec:nonHermSoS}
It seems necessary that we have a way to do non-commutative polynomial optimization with non-Hermitian operators and, in particular, a non-Hermitian objective. Typically the NPA hierarchy focuses entirely on Hermitian operators and Hermitian polynomial constraints (equations whose form is manifestly invariant under the adjoint). Follow up works\supercite{NPA10_nonherm} allowed non-Hermitian variables, but the constraint equations and objective had to be Hermitian polynomials.

It turns out that we can express {\em general} non-commutative polynomial optimization problem, with non-Hermitian variables, constraints, and objective, can be written in entirely Hermitian terms. This would be necessary if one wanted to attempt the NPO relaxation for reversible capacity in Appendix \ref{sec:revNPO}.

Any operator $X$ can be split into its Hermitian and anti-Hermitian parts: $X = X_h + i X_a$. Then $X_h$ and $X_a$ are two Hermitian variables. Our approach is to express any problem in terms of these parts. Scalars and sums go through in the straightforward way:
$$\alpha\in \mathbb{R}:\quad \alpha X = (\alpha X_h) + i (\alpha X_a)$$
$$ i X = (-X_a) + i(X_h)$$
$$X + Y = (X_h+Y_h) + i (X_a+Y_a)$$
For a product $X Y$, the Hermitian part can be computed as $\frac{XY + Y^\dagger X^\dagger}{2}$, and the anti-Hermitian part is what remains.
$$(XY)_h = \frac{1}{2}((X_h Y_h + Y_h X_h) - (X_a Y_a + Y_a X_a) + i (X_h Y_a - Y_a X_h) + i (X_a Y_h - Y_h X_a))$$
$$(XY)_a = \frac{1}{2}((X_h Y_a + Y_a X_h) + (X_a Y_h + Y_h X_a) -i(X_h Y_h - Y_h X_h) + i(X_a Y_a - Y_a X_a))$$
The equivalence goes through both ways: any valid assignment of $X_h$ and $X_a$ that obeys the constraints gives a corresponding $X$, and vice versa. This is enough to then write any polynomials involving these Hermitian variables, and all expressions generated are manifestly Hermitian polynomials. If any one of our variables $H$ is already Hermitian, then we just use $H = H_h$ and omit a variable $H_a$.

Then there is the objective, which would generally take the form $\lambda_{\max}(M)$. For this to be a meaningful ``max", we define this to be the largest {\em real, right} eigenvalue of the operator. If $M$ is Hermitian, this fits in the standard framework. If $M$ is not Hermitian, add a new Hermitian variable $P$, the constraint $MP = P^2$, and the objective $\lambda_{\max}(P)$. We claim that this gives an equivalent description of the problem as long as $\lambda_{\max}(M) \ge 0$.
\begin{enumerate}
    \item For any feasible solution to the original problem, let $\lambda$ be the maximum real eigenvalue of $M$, with associated eigenvector $v$. Let $P = \lambda vv^\dagger$. Then
$$MP = M (\lambda v v^\dagger) = \lambda (Mv) v^\dagger = \lambda (\lambda v) v^\dagger = \lambda^2 v v^\dagger = P^2$$
so the new problem has at least as good an optimum.
    \item For the other direction, take a feasible solution to the new problem with a score $\lambda = \lambda_{\max}(P) > 0$ and an associated eigenvector $v$. Since we know that $MP = P^2$, it must be that
$$Mv = M(\lambda^{-1}Pv) = \lambda^{-1}MPv = \lambda^{-1}P^2v = \lambda v$$
and $v$ is a right eigenvector of $M$ with eigenvalue $\lambda$, as long as $\lambda \neq 0$. So the original problem (without $P$) has a feasible solution where $M$ has an eigenvalue of $\lambda$ as well.
\end{enumerate}
If $\lambda_{\max}(M) \not\ge 0$ there is an issue. $P=0$ is always a solution as well, with an according objective value of $\lambda_{\max}(0) = 0$, so if the original problem had a negative optimum value, then the new problem will have an optimum eigenvalue of zero. This isn't a problem in practice: as long as we have some lower bound on the original problem, that the optimum is at least $-k$, we can replace the objective $M$ with $M+kI$ to raise the optimum value above zero.

\begin{theorem}
    Any noncommutative polynomial eigenvalue optimization problem with objective as the largest real eigenvalue, and some known lower bound on the optimum value, can be rewritten purely in terms of Hermitian operators and Hermitian polynomials, with at most $2n+1$ variables and $2c+2$ constraints.
\end{theorem}

\end{document}